\documentclass[10pt,conference,letterpaper,twocolumn]{IEEEtran}
\usepackage[dvips]{graphicx}
\usepackage{cite}
\usepackage{epsfig}
\usepackage{amsmath,amssymb,amsfonts,amstext,amsbsy,amsopn,dsfont}
\usepackage{cases}
\usepackage{sublabel}
\usepackage{array}

\newtheorem{theorem}{\textbf{Theorem}}

\newtheorem{lemma}{\textbf{Lemma}}
\newtheorem{remark}{\textbf{Remark}}
\newtheorem{definition}{\textbf{Definition}}

\newcommand{\defn}{\triangleq}
\newcommand{\dif}{\textmd{d}}
\newcommand{\ie}{i.e., }

\begin{document}

\title{Multicast Capacity Scaling of Wireless Networks with Multicast Outage}

\author{Chun-Hung Liu and Jeffrey G. Andrews \\Department of Electrical and Computer Engineering \\The University of Texas at Austin, Austin TX 78712-0204, USA\\ Email: chliu@mail.utexas.edu and jandrews@ece.utexas.edu}

\maketitle

\begin{abstract}
Multicast transmission has several distinctive traits as opposed to more commonly studied unicast networks.  Specially, these include (i) identical packets must be delivered successfully to several nodes, (ii) outage could simultaneously happen at different receivers, and (iii) the multicast rate is dominated by the receiver with the weakest link in order to minimize outage and retransmission. To capture these key traits, we utilize a Poisson cluster process consisting of a distinct Poisson point process (PPP) for the transmitters and receivers, and then define the multicast transmission capacity (MTC) as the maximum achievable multicast rate times the number of multicast clusters per unit volume, accounting for outages and retransmissions. Our main result shows that if $\tau$ transmission attempts are allowed in a multicast cluster, the MTC is $\Theta\left(\rho k^{x}\log(k)\right)$ where $\rho$ and $x$ are functions of $\tau$ depending on the network size and density, and $k$ is the average number of the intended receivers in a cluster. We also show that an appropriate number of retransmissions can significantly enhance the MTC.
\end{abstract}

\section{Introduction}

Multicast refers to the scenario whereby a transmitter needs to send a packet to multiple receivers.  In a wireless network, this creates a two-edged sword. On one hand, the broadcast nature of wireless transmission assists multicast; but roughly uncorrelated outage probabilities at each receiver (due to spatially distinct fading and interference) require retransmissions that cause interference and waste.  Multicast is an important aspect of sensor and tactical networks, and increasingly in commercial networks where streaming is supported.  However, the literature on multicast is minuscule compared to unicast. In this work we attempt to investigate the fundamental throughput limits of multicast transmission and we develop a metric based on spatial outage capacity which we term \emph{multicast transmission capacity} (MTC).

In order to characterize the MTC in a wireless network we propose a multicast network model in which each transmitter has an intended multicast region (called a cluster) where all the intended receivers are uniformly and independently scattered, and hence a Poisson cluster process can be reasonably used to model the transmit-receiver location statistics. The active transmitters are modeled as a stationary Poisson point process (PPP) and their associated receiver nodes in the cluster are also a stationary PPP, as shown in Fig. \ref{Fig:MulticastModel}. In other words, each cluster is randomly located in the network and comprises a multicast session.  This paper will develop interference and outage expressions for this model, and analyze some important cases of the network model and design space.

\begin{figure}[h]
  \centering
  \includegraphics[width=3.0in,height=1.8in]{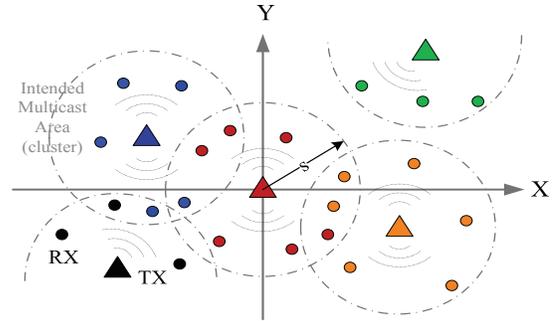}\\
  \caption{Multicast transmission model in a 2-dimensional ad hoc network: Each transmitter (triangle) and its intended receivers (small circles) are indicated by the same color in a cluster of radius $s$.}
  \label{Fig:MulticastModel}
\end{figure}

%\subsection{Motivation and Related Work}

Some previous works, such as \cite{SSXLRS07,XYL09,AKVRRR06,PJGR05,AKHVRRR06,AKHRR07}, have made significant progress in studying the multicast or broadcast capacity. For example, in \cite{SSXLRS07} the protocol model is used where source nodes and their multicast destinations are randomly chosen. The multicast capacity is defined as the sum rate of all multicast flows and it is obtained as a function of the number of multicast sources. In \cite{XYL09}, the multicast capacity under the protocol model is defined as the transmission rate summed over all of the multicast traffic flows in the network. Its scaling characterization is obtained by the number of receivers in each multicast session. In \cite{RZHENG06}, the physical model and a stationary PPP of the nodes in the network are considered. It showed that the broadcast capacity is a constant factor of the computed upper bound when the number of nodes goes to infinity under a constant node intensity.

The multicast and broadcast capacities in the prior works are not investigated from the multi-receiver outage point of view and thus their scaling results cannot provide us retransmission guidelines for outage reduction and capacity enhancement. The multicast capacity problem here is studied from a multicast outage perspective. The main issues we would like to clarify are that how many multicast sessions which interfere each other can coexist if all receivers in a multicast session need to receive the packet from their transmitter, and when retransmission is beneficial or detrimental to the multicast capacity.

In this paper, we introduce the MTC in a $d$-dimensional network, which is defined as the maximum achievable multicast rate per transmission attempt times the maximum number of the coexisting clusters in the network per unit volume subject to decoding delay and multicast outage constraints. The idea of defining MTC is from the transmission capacity framework originated in \cite{SWXYJGAGDV05}. \emph{The decoding delay constraint here means a transmitter can multicast a packet to all of its intended receivers up to $\tau\in\mathbb{N}_+$ transmission attempts and multicast outage happens when any of the intended receivers in a cluster does not receive the information multicasted by their transmitter during $\tau$ attempts.}

\textbf{Main Contributions}. In this work our first contribution is introducing a first way to define the MTC with multicast outage, and propose some new cluster-based definitions for the largeness and denseness of a network in order to characterize the scaling behaviors of the MTC under different conditions. We show that the MTC scaling can be expressed in a general form of $\Theta(\rho k^{x}\log(k))$ where $\rho$ and $x$ are given in Table \ref{Tab:MainResultsMTC} for different network conditions, $k$ is the average number of the intended receivers in a cluster. From the scaling characterizations, we know retransmissions have a significant effect on the MTC and certain number of retransmissions could enhance it. In addition, we also show that the MTC scaling is not affected by the fading model of communication channels.

\begin{table}[h]
\centering
\caption{Main Results on Multicast Transmission Capacity}\label{Tab:MainResultsMTC}
\begin{tabular}{|c|c|c|c|}
\hline
\bf MTC & \multicolumn{3}{c|} {$\Theta(\rho\,k^{x}\log(k))$}\\
\hline
\bf Network Condition & Dense & Large  & Large Dense \\
\hline
$x$ & $-\frac{1}{\tau}$ & $-\left(1+\frac{1}{\tau}\right)$ & $-\left(\frac{\tau+2}{2\tau}\right)$\\
\hline
$\rho$ & \multicolumn{3}{c|} {$\frac{1}{\tau^2}\sqrt[\tau]{\epsilon(\tau+1)}$}\\
\hline
\end{tabular}
\end{table}

\section{Network Model and Preliminaries}\label{Sec:MTCmodelPrelims}

The multicast network model should be with tractability of finding the necessary information for MTC. The following Poisson cluster model is proposed because it not only captures the multicast behavior in the network but also makes all derivations easier.

\subsection{Multicast Transmission Model}\label{Sec:MultiTranModelRxConnProc}
In the network, each transmitter has a multicast cluster of equal volume and its receive nodes in the cluster suffer aggregate interference from a Poisson field of transmitters. Specifically, we assume that the network is operating a slotted ALOHA protocol and the distribution of the transmitting nodes in the network is a stationary Poisson point process (PPP) $\Phi_t$ of intensity $\lambda_t$. As shown in Fig. \ref{Fig:MulticastModel}, any transmitter $X_i\in\Phi_t$ has its own intended multicast cluster $\mathcal{R}_i$ where all of its intended receivers are uniformly and independently distributed and they also form a stationary PPP $\Phi_{r_i}$ of intensity $\lambda_r$. \emph{Note that each cluster could contain other transmitters and unintended receivers except its own transmitter and intended receivers.}

Accordingly, the multicast transmission sessions in the network are essentially the Poisson cluster process $Z_i\defn \Phi_{r_i}\cup X_i$, \ie each transmitter is a \emph{parent} node associated with a cluster of receive \emph{daughter} nodes. The cluster processes $\{Z_{i}\}$ corresponding to different transmitters $\{X_i\}$ are assumed to be independent so that the superposition of all clusters yields the resulting cluster process $\Phi=\bigcup_{X_i\in\Phi_t}Z_i$ of intensity $\lambda=k\,\lambda_t$ where $k=\mu_r\,\lambda_r$ is the average number of the intended receivers in each cluster assuming all $\{\mathcal{R}_i, \forall i\in \mathbb{N}\}$ have the same volume $\mu_r$\footnote{$\mu_r$ is the Lebesgue measure of $\mathcal{R}_i$. For example, if $\mathcal{R}_i$ is a \emph{d}-dimensional ball of radius $s$ then $\mu_r=\mu_u\,s^d$, where $\mu_u =\sqrt{\pi^d}/\Gamma(1+d/2)$ is the Lebesgue measure of a \emph{d}-dimensional unit ball\cite{DSWKJM96}.}. The distribution of the intended receiver nodes in each cluster is assumed as a \emph{marked} PPP denoted by $\Phi_{r_i} \defn \{(Y_{ij},H_{ij}): Y_{ij}\in \mathcal{R}_i,j\in\mathbb{N}\}$, where $H_{ij}$ is the fading channel gain between transmitter $X_i$ and its intended receiver $Y_{ij}$. Similarly, the distribution of the transmitters in the network is also a marked PPP, \ie $\Phi_t \defn \{(X_i,\{\tilde{H}_{ij}\}), i,j\in \mathbb{N}\}$ where $\tilde{H}_{ij}$ denotes the fading channel gain between transmitter $X_i$ and the receiver $Y_{0j}$ located in cluster $\mathcal{R}_0$. All the fading channel gains are i.i.d. with  probability density function (PDF) $f_{H}(h)$.

Without loss of generality, the MTC can be evaluated in the reference cluster $\mathcal{R}_0$ whose transmitter $X_0$ is located at the origin. We condition on this typical transmitter $X_0$ resulting in what is known the Palm distribution for transmitting nodes in the $d$-dimensional Euclidean space \cite{DSWKJM96}. It follows by Slivnyak's theorem \cite{DSWKJM96} that this conditional distribution also corresponds to a homogenous PPP with the same intensity and an additional point at the origin. The signal propagation in space is assumed to undergo path loss and fading. The path loss model between two nodes $X$ and $Y$ used in this paper is
\begin{equation}\label{Eqn:PathLossModel}
\|X-Y\|^{-\alpha}\defn\begin{cases}|X-Y|^{-\alpha},&\quad \text{if}\,\, |X-Y|\geq 1 \\ 0, &\quad \text{otherwise}, \end{cases}
\end{equation}
where $|X-Y|$ denotes the Euclidean distance between nodes $X$ and $Y$, and $\alpha>d$ is the path loss exponent.
The reason of using the model in \eqref{Eqn:PathLossModel} is because the model $|\cdot|^{-\alpha}$ does not behave well in the near field of each transmitter and it thus leads to an unbounded mean of the shot noise process. The Nakagami-$m$ fading model\footnote{In this paper, the Nakagami-$m$ fading channels are always assumed to have unit mean, unit variance and $m\in\mathbb{N}_+$.} is adopted here because it covers several different fading models, such as Rayleigh fading, Rician fading and no fading, etc. By using this model, we can observe if different fading models affect MTC or not.

The intended multicast region $\mathcal{R}_i$ of transmitter $X_i$ is confined by $\mathcal{B}(X_i,s)$ which is a $d$-dimensional ball centered at transmitter $X_i$ with radius $s \geq 1$. Each receiver in $\Phi_r$\footnote{Here $\Phi_r$ means $\Phi_{r_0}$. The subsubscript ``0'' of $\Phi_{r_0}$ is dropped for notation simplification. Since all of the following analysis is based on the nodes in the reference cluster $\mathcal{R}_0$, the subscript or subsubscript 0 of some variables will not be explicitly indicated if there is no ambiguity. For instance, $Y_j$ and $H_j$ in $\mathcal{R}_0$ actually stand for $Y_{0j}$ and $H_{0j}$, respectively.} is able to successfully receive the message if its SIR is greater or equal to the target threshold $\beta$. That is, receiver node $Y_j\in \Phi_r$ is ``connected'' to the typical transmitter if
\begin{equation}\label{Eqn:SINRwThreshold}
\frac{\,H_j \|Y_j\|^{-\alpha}}{\,I_t}\geq \beta,
\end{equation}
where all the transmitters are assumed to use the same transmit power, the network is interference-limited, and $I_t$ is the aggregate interference at receive node $Y_j$ and a sum over the marked point processes. Namely,
\begin{equation}\label{Eqn:PoissonShotNoise1}
    I_t=\sum_{X_{i}\in\Phi_t\setminus\{X_0\}} \tilde{H}_{ij}\|X_{i}-Y_j\|^{-\alpha},
\end{equation}
which is a Poisson shot noise process, and $\tilde{H}_{ij}$ is the fading channel gain from transmitter $X_i$ to receiver $Y_j$ in $\mathcal{R}_0$. Since $\Phi_t$ is stationary,  according to Slivnyak's theorem the statistics of signal reception seen by receiver $Y_j$ is the same as that seen by any other receivers in the same cluster. Thus $I_t$ can be evaluated at the origin, \ie \eqref{Eqn:PoissonShotNoise1} can be rewritten as $I_t=\sum_{X_i\in\Phi_t\setminus\{X_0\}}\tilde{H}_i \|X_i\|^{-\alpha}.$

Suppose the decoding delay is up to the lapse of $\tau$ transmission attempts for a transmitter. The connected receiver process $\hat{\Phi}_{c_i}$ for the $i$-th transmission is denoted by
\begin{equation}
\hat{\Phi}_{c_i} = \left\{(Y_j,H_{j_i})\in\Phi_r:  H_{j_i}\geq\beta\|Y_j\|^{\alpha}I_t\right\},
\end{equation}
where $\{H_{j_i}\}$ are i.i.d. for all $i\in [1,\cdots,\tau]$. Also, let $\Phi_c$ be the connected receiver process at the $\tau$th attempt, \ie it is the set of all intended receivers in a cluster connected by their transmitter during the decoding delay, and thus it can be written as $\Phi_c = \bigcup_{i=1}^{\tau} \hat{\Phi}_{c_i}$. In other words, the connected receiver process can be described by a \emph{filtration} process\footnote{A filtration process means $\Phi_{c_1}\subseteq \Phi_{c_2}\cdots\subseteq \Phi_{c_j}$, and for any set $\mathcal{A}\subseteq \mathcal{R}_0$, $\Phi_{c_j}(\mathcal{A})\rightarrow \Phi_r(\mathcal{A})$ almost surely as $j\rightarrow \infty$ where $\Phi(\mathcal{A})$ denotes the random number of point process $\Phi$ enclosed in set $\mathcal{A}$.}.

\subsection{Multicast Transmission Outage}\label{Sec:MultiTransOutage}
The transmission capacity of an ad hoc network introduced in \cite{SWXYJGAGDV05} is defined based on point-to-point transmission with an outage probability constraint $\epsilon\in(0,1)$, and is given by
\begin{equation}\label{Eqn:UnicastTC}
c_{\epsilon} = b\,\bar{\lambda}_t\,(1-\epsilon),
\end{equation}
where $b$ is the constant transmission rate a communication link can support (for example, about $\log_2(1+\beta)$), and $\bar{\lambda}_t$ is the maximum contention intensity subject to an outage probability target $\epsilon$. However, \eqref{Eqn:UnicastTC} cannot be directly applied to multicast because the multicast rate will be affected by $\mu_r$ and $\bar{\lambda}_t$, and the outage of a multicast transmission is not point-to-point but \emph{point-to-multipoint}.

Since no desired receiver can be assumed to be dispensable, a reasonable way to define outage is when \emph{any of the intended receivers of a transmitter does not receive a multicasted packet} during a period of time up to the decoding delay. Thus, a multicast outage event of each multicast cluster can be described as $\mathcal{E} =\{\Phi_c\subset\Phi_r\}$. The probability of $\mathcal{E}$ can be characterized by the intensity of the connected receivers during the lapse of $\tau$ attempts as follows:
\begin{eqnarray}\label{Eqn:MultiOutProb1}
\mathbb{P}[\mathcal{E}] &\defn& 1-\mathbb{P}[\{\Phi_r\setminus\Phi_c\}=\emptyset]\nonumber\\
&=& 1-\exp\left\{-\int_{\mathcal{R}_0}(\lambda_r-\lambda_c(Y,\tau))\,\mu(\dif Y)\right\},
\end{eqnarray}
where $\mu$ is a $d$-dimensional Lebesgue measure. Since all of the intended receivers are uniformly distributed in $\mathcal{R}_0$, \eqref{Eqn:MultiOutProb1} can be rewritten as
\begin{equation}\label{Eqn:MultiOutProb2}
\mathbb{P}[\mathcal{E}]=
1-\exp\left\{-\mu_r\,(\lambda_r-\mathbb{E}_{R}[\lambda_c(R,\tau)])\right\}\leq \epsilon,
\end{equation}
where $R\in[0,s]$ is a random variable whose PDF is $f_R(r)=d\,\frac{r^{d-1}}{s^d}$. The outage probability in \eqref{Eqn:MultiOutProb2} cannot exceed its designated upper bound $\epsilon$ which is usually a small value.

\begin{remark}
When a transmitter has a fixed number of the intended receivers in a cluster, the probability of multicast outage based on our definition is too complex to be calculated due to the spacial correlation of interference \cite{RKGMH09}. Since $\Phi_r$ is a stationary PPP and $\Phi_c$ is a nonhomogeneous PPP (this point will be proved in Section \ref{Sec:RxConnectedPro}), we can obtain the result in \eqref{Eqn:MultiOutProb1} so that the multicast outage probability can be easily carried out via finding $\mathbb{E}_R[\lambda_c(R,\tau)]$ in \eqref{Eqn:MultiOutProb2}. Since $\mathbb{E}_R[\lambda_c(R,\tau)]$ is a monotonically decreasing function of $\lambda_t$, $\bar{\lambda}_t$ can be reached by the lower bound on $\mathbb{E}_R[\lambda_c(R,\tau)]$ obtained from \eqref{Eqn:MultiOutProb2}.
\end{remark}

\subsection{Definitions of MTC, Largeness and Denseness of Networks}\label{Sec:DefnMTC}
The MTC defined in the following is similar to the idea of the transmission capacity defined in \eqref{Eqn:UnicastTC}, but there are some subtle differences. Since the multicast outage probability is upper bounded by a small $\epsilon$, a Taylor expansion gives $\bar{\lambda}_t(\epsilon)=\bar{\lambda}_{\epsilon}(\epsilon)+O(\epsilon^2)$, where $\bar{\lambda}_{\epsilon}$ is a reasonable approximation of $\bar{\lambda}_t$ for small $\epsilon$. For simplicity, we will focus the analysis on $\bar{\lambda}_{\epsilon}$. In addition, the transmission rate $b$ in \eqref{Eqn:UnicastTC} becomes a \emph{multicast rate}, which is not necessary equal to a constant.
\begin{definition}\label{Def:MTC}
The multicast transmission capacity with the multicast outage probability defined in \eqref{Eqn:MultiOutProb1} for small $\epsilon$ is defined as
\begin{equation}\label{Eqn:MTCwoACK}
 C_{\epsilon} \defn \frac{1}{\tau}\,b\,\bar{\lambda}_{\epsilon}\,(1-\epsilon),
\end{equation}
where $b$ is the maximum achievable multicast rate on average for every cluster.
\end{definition}

$C_{\epsilon}$ essentially gives area spectral efficiency of cluster-based multicast transmission. The following definitions of largeness and denseness of a network will be needed to acquire the scaling characterizations of the MTCs in the subsequent analysis. They are defined based on the large average number of the intended receivers $k$ in a cluster.
\begin{definition}\label{Def:DenseLargeNetwork}
(a) We say a network is ``large'' if the volume $\mu_r$ of a cluster in the network is sufficiently large such that for fixed $\lambda_r$ we have $k\gg 1$. (b) If the intended receiver intensity is sufficiently large such that for fixed volume $\mu_r$ we have $k\gg 1$, then such a network is called ``dense''. (c) A ``large dense'' network, it means that clusters in a network have a sufficiently large size as well as receiver intensity; namely, $\lambda_r\propto \mu_r$ and thus $k \gg 1$.
\end{definition}

%\section{Multicast Transmission Capacity with Single-Hop Multicast}\label{Sec:MTCwSinglehop}

\section{The Receiver-Connected Point Process}\label{Sec:RxConnectedPro}
During the allowed $\tau$ transmission attempts, the intended receivers in $\mathcal{R}_0$ connected by transmitter $X_0$ form a receiver-connected point process whose intensity is the necessary information to estimate the multicast outage probability. Since $\Phi_c$ is a filtration process and upper bounded by $\Phi_r$ (as explained in Section \ref{Sec:MultiTranModelRxConnProc}), the connected receiver intensity in $\mathcal{R}_0$ is an increasing function of $\tau$ as shown in the following lemma.
\begin{lemma}\label{Lem:NonCoopConnRxIntenNakaFading}
Consider the stationary PPP in the reference cluster $\mathcal{R}_0$. If a transmitter is allowed to transmit a packet up to $\tau$ times and all channels are Nakagami-$m$ fading, then $\Phi_c$ is a nonhomogeneous thinning PPP and its intensity at $r\in[1,s]$ is lower bounded as follows.
\begin{eqnarray}\label{Eqn:NonCoopConnIntenNakaFading}
\lambda_{c}(r,\tau) &\geq& \lambda_r\left\{1-\left[1-(\beta\,r^{\alpha})\Psi^{(m-1)}(\beta\, r^{\alpha})\right]^{\tau}\right\},\\
\Psi^{(m)}(\phi) &\defn& \frac{(-\phi)^{m}}{m!}\,\frac{\dif^{m}}{\dif \label{Eqn:PsiMphi} \phi^{m}} \left(e^{-\mu_u\,\lambda_t\,\Delta_1(\phi,\infty)}/\phi\right),
\end{eqnarray}
\end{lemma}
where $\Delta_1$ is defined in the following lemma for moment generating functional of stationary independent PPPs.
\begin{proof}
See Part \ref{App:ProofNonCoopConnRxInten} in the Appendix.
\end{proof}

\begin{lemma}\label{Lem:LapalceShotPPP}
Let $\Phi_i=\{(X_{i_j},H_{i_j}):X_{i_j}\in\mathcal{B}(0,r)\cap\mathbb{R}^d, r\geq 1,j\in\mathbb{N}\}$ be a stationary marked PPP of intensity $\lambda_i$ for all $i\in[1,2,\cdots,\ell]$ and $\{H_{i_j}\}$ are i.i.d. Nakagami-$m$ random variables with unit mean and variance. Suppose $\Phi_i$ has a Poisson shot generating function $I_i:\mathbb{R}_+^d\times \mathbb{R}_+\rightarrow \mathbb{R}_+$ which is defined as $I_i\defn \sum_{X_i\in\Phi_i} H_{i_j}\|X_{i_j}\|^{-\alpha}$ where $\alpha>d$. If $\{\Phi_i\}$ are independent and $\xi=\frac{d}{\alpha}$, then the sum of the Poisson shot generating functions, \ie $I=\sum_{i=1}^{\ell}I_i$, has the following moment generating functional for $\lambda=\sum_{i=1}^{\ell}\lambda_i$, $\phi_1\in\mathbb{R}_{++}$ and $\phi_2\in\left(0,m\,r^{\alpha}\right)$:
\begin{eqnarray}
\mathcal{L}_I(\phi_1)&=&\mathbb{E}\left[e^{-\phi_1 I}\right]= \exp\left(-\mu_u\,\Delta_1(\phi_1,r)\,\lambda\right), \label{LaplaceShotPPP1}\\
\mathcal{M}_I(\phi_2)&=&\mathbb{E}\left[e^{\phi_2 I}\right]= \exp\left(\mu_u\,\Delta_2(\phi_2,r)\,\lambda\right), \label{LaplaceShotPPP2}
\end{eqnarray}
where
\begin{eqnarray*}
\Delta_1(\phi_1,r)&\defn& \xi\left(\frac{\phi_1}{m}\right)^{\xi}\sum_{j=0}^{m-1} {m\choose j}\int_{m/\phi_1}^{m r^{\alpha}/\phi_1} \frac{t^{j+\xi-1}}{(1+t)^m}\dif t,\\
\Delta_2(\phi_2,r) &\defn& \xi\left(\frac{\phi_2}{m}\right)^{\xi}\sum_{j=0}^{m-1} {m\choose j} \int_{m/\phi_2}^{m r^{\alpha}/\phi_2}\frac{(-1)^j t^{j+\xi}}{(1-t)^m}\,\dif t.
\end{eqnarray*}
\end{lemma}
\begin{proof}
Omitted due to space. See \cite{CHLJGA10} if interested.
\end{proof}

\section{Multicast Transmission Capacity}\label{Sec:NonCoopMTC}
In this section, we study the MTC when transmitters are multicasting to all of their intended receivers in a single-hop fashion. Using the receiver-connected intensity found in Section \ref{Sec:RxConnectedPro} and the multicast outage probability defined in \eqref{Eqn:MultiOutProb2}, the maximum contention intensity is found as follows.

\begin{theorem}\label{Thm:MaxContenIntenNonCoop}
Suppose \eqref{Eqn:MultiOutProb2} is upper bounded by $\epsilon$ and the maximum decoding delay is $\tau$ transmission attempts. If $k\geq \frac{1}{\epsilon^{\tau-1}}$, then the maximum contention intensity is
\begin{equation} \label{Eqn:MaxContenIntenNoCoop}
\bar{\lambda}_{\epsilon}=\Theta\left(\frac{\rho\,\tau^2}{\beta^{\xi}\,\mu_r\,\sqrt[\tau]{k}}\right),
\end{equation}
where $\rho=\frac{1}{\tau^2}\sqrt[\tau]{\epsilon(\tau+1)}$.
\end{theorem}
\begin{proof}
See Part \ref{App:ProofMaxContenIntenNonCoop} in the Appendix.
\end{proof}

\begin{remark}
The scaling function $\Theta(\cdot)$ of $\bar{\lambda}_{\epsilon}$ in \eqref{Eqn:MaxContenIntenNoCoop} only contains the ``controllable'' network parameters such as $\mu_r$, $\lambda_r$, $\epsilon$ and $\beta$, which means their values are adjustable if needed. The parameter $m$ of Nakagami fading, is a channel characteristic and usually uncontrollable and thus not involved in the scaling result in Theorem \ref{Thm:MaxContenIntenNonCoop}. Hence, the scaling of the maximum contention intensity is unaltered no matter whether the channels in the network are fading or not.
\end{remark}

If a unicast planar network without retransmission is considered (\ie $d=2$ and $k=\tau=1$), $\bar{\lambda}_{\epsilon}$ in \eqref{Eqn:MaxContenIntenNoCoop} will reduce to the previous results discovered, \ie $\bar{\lambda}_{\epsilon}=\Theta\left(\frac{\epsilon}{s^2\,\beta^{2/\alpha}}\right)$. In \cite{SWXYJGAGDV05}, for example, the maximum contention intensities of FH-CDMA and DS-CDMA are $\Theta\left(\frac{\epsilon M}{s^2\beta^{2/\alpha}}\right)$ and $\Theta\left(\frac{\epsilon M^{2/\alpha}}{s^2\beta^{2/\alpha}}\right)$ respectively, where $M$ is the channel number of FH-CDMA and the spreading factor of DS-CDMA. It is easy to check that these two results coincide with ours here by considering $\frac{\bar{\lambda}_{\epsilon}}{M}$ for FH-CDMA and $\frac{\beta}{M}$ for DS-CDMA. In addition, the longest transmission distance in a cluster is $s$ and we know $\bar{\lambda}_{\epsilon}=\Theta(s^{-2})$ and so is the network capacity, which also coincides with the results in \cite{SWXYJGAGDV05}\cite{FBBBPM06}.

Considering a capacity-approaching code is used, the maximum achievable multicast rate $b$ that is acceptable for all intended receivers is the following Ergodic channel capacity evaluated at the boundary of a cluster:
\begin{equation}\label{Eqn:MaxBcRateWoTimeDiv}
b = \mathbb{E}\left[\log \left(1+\frac{H_{\max}\,s^{-\alpha}}{I_t}\right)\right].
\end{equation}
where $H_{\max}= \max_{i\in[1,\cdots,\tau]}H_{j_i}$.
The bounds on $b$ are given in the following lemma.
\begin{lemma}\label{Lem:NonCoopBoundsMCrate}
There exists a $\delta\in(0,1)$ such that the bounds on the multicast rate $b$ in \eqref{Eqn:MaxBcRateWoTimeDiv} can be given by
\begin{equation}\label{Eqn:BoundsNonCoopBCRate}
\delta\,\log\left(1+\frac{1}{\mu_r\,\lambda_t}\right) \leq b\leq\log\left(1+\frac{1}{\mu_r\lambda_t}\right)+O(1).
\end{equation}
\end{lemma}
\begin{proof}
Omitted due to space. See \cite{CHLJGA10} if interested.
\end{proof}

\begin{remark}
The bounds on $b$ in \eqref{Eqn:BoundsNonCoopBCRate} are not affected by channel fading since the fading effect has been averaged out, and they indicate $b$ is significantly dominated by the aggregate interference from the transmitters in the cluster.
\end{remark}

\textbf{Scaling Law of MTC}. According to Theorem \ref{Thm:MaxContenIntenNonCoop} and Lemma \ref{Lem:NonCoopBoundsMCrate}, we found that the multicast rate $b$ is $\Theta(\log(k)/\tau)$ for any network conditions if $\bar{\lambda}_{\epsilon}$ is achieved and $k$ is sufficiently large. The MTCs in a network without receiver cooperation can be concluded as follows. (\textbf{i}) For a dense network, $\bar{\lambda}_{\epsilon}=\Theta\left(\frac{\rho\tau^2}{\sqrt[\tau]{k}}\right)$ since $\mu_r$ is fixed and $k\gg 1$. By the MTC definition, we know $C_{\epsilon}=\Theta\left(\frac{\rho\log(k)}{\sqrt[\tau]{k}}\right)$. (\textbf{ii}) If the network is large, then $\bar{\lambda}_{\epsilon}=\Theta\left(\frac{\rho\tau^2}{k^{1+1/\tau}}\right)$ and thus $C_{\epsilon}=\Theta\left(\frac{\rho\log(k)}{k^{(1+1/\tau)}}\right)$. (\textbf{iii}) For a large dense network, $\lambda_r=\Theta(\mu_r)$ and $k\gg 1$. So $\bar{\lambda}_{\epsilon}$ is $\Theta\left(\rho\,\tau^2\,k^{-\left(\frac{\tau+2}{2\tau}\right)}\right)$ and thus $C_{\epsilon}=\Theta\left(\rho k^{-\left(\frac{\tau+2}{2\tau}\right)}\log(k)\right)$. In summary, the MTC here can be expressed in a general form as follows:
\begin{equation}\label{Eqn:GeneralMTC}
 C_{\epsilon} = \Theta\left(\rho\,k^{x}\log(k)\right),
\end{equation}
where $x$ has been given in Table \ref{Tab:MainResultsMTC}.

%----------------------------------Section of Appendices -------------------------------------------

%\section{Appendix}
\appendix
\subsection{Proof of Lemma \ref{Lem:NonCoopConnRxIntenNakaFading}}\label{App:ProofNonCoopConnRxInten}
The Laplace functional of the stationary PPP $\Phi$ for a nonnegative function $g:\mathbb{R}^d\rightarrow\mathbb{R}_+$ is defined and shown as follows \cite{FBBB10}:
\begin{eqnarray*}
 \mathcal{\tilde{L}}_{\Phi}(g)&\defn& \mathbb{E}\left[e^{-\int_{\mathbb{R}^d}g(X)\,\Phi(\dif X)}\right]\\
 &=&\exp\left(-\int_{\mathbb{R}^d}\left(1-e^{-g(X)}\right)\lambda_r\,\mu(\dif X)\right).
\end{eqnarray*}
Since the Laplace functional completely characterizes the distribution of the point process, we can find the intensity of $\Phi_c$ by looking for $\mathcal{\tilde{L}}_{\Phi_c}(g)$. Recall that $\Phi_c=\left\{Y_j\in\Phi_r: \max_{i\in [1,\cdots,\tau]} H_{j_i} \geq \beta\|Y_j\|^{\alpha}I_t\right\}$ and $\{H_{j_i}\}$ are i.i.d. $\forall i\in [1,2,\cdots,\tau]$. Let $\mathds{1}_A(x)$ be an indicator function which is equal to 1 if $x\in A$ and 0, otherwise. The Laplace functional of $\Phi_c$ for $g(Y)=\tilde{g}(Y)\mathds{1}_{\Phi_c}(Y)$ is given by \eqref{Eqn:LapFunG}.
\begin{eqnarray}
\mathcal{\tilde{L}}_{\Phi_c}(g) = e^{-k} \sum_{i=0}^{\infty} \frac{\lambda^i_r}{i!}\int_{\mathcal{R}_0}\cdots\int_{\mathcal{R}_0}
\prod_{j=1}^i \bigg[1-\mathbb{E}[\mathds{1}_{\Phi_c}(Y_j)]\nonumber\\
\left(1-e^{-g(Y)}\right)\bigg]\mu(\dif Y_1)\cdots \mu(\dif Y_i),\nonumber\\
%&=& e^{-k} \sum_{i=0}^{\infty} \frac{1}{i!}\left(\int_{\mathcal{R}_0}\left(e^{-g(Y)}\mathbb{P}[Y\in\Phi_c]+1-\mathbb{P}[Y\in\Phi_c]\right)\lambda_r\,\mu(\dif Y)\right)^i\\
= \exp\left(\int_{\mathcal{R}_0}\left(e^{-g(Y)}-1\right)\mathbb{E}[\mathds{1}_{\Phi_c}(Y)]\lambda_r\,\mu(\dif Y)\right).\label{Eqn:LapFunG}
\end{eqnarray}
Also, for all $r\in[1,s]$ we know
\begin{eqnarray*}
\mathbb{E}[\mathds{1}_{\Phi_c}(Y)]&=&\mathbb{P}\left[\max_{i=1,\ldots,\tau}H_{j_i}\geq \beta \|Y\|^{\alpha}I_t\right]\\
&\stackrel{(\star)}{=}& 1-\left(\mathbb{E}_{I}[F_H(\beta \, r^{\alpha}\,I_t|I_t)]\right)^{\tau},\label{Eqn:ConnProbNtrans1}
\end{eqnarray*}
where $(\star)$ follows from the fact that the temporal correlation of interference can be neglected for small $\lambda_t$\cite{RKGMH09}.
So we have
\begin{equation*}
\mathcal{\tilde{L}}_{\Phi_c}(g)= \exp\left(-\int_1^s \left(1-e^{-g(r)}\right)\lambda_c(r,\tau)\,\mu(\dif r)\right),
\end{equation*}
where $\lambda_c(r,\tau)=\lambda_r(1-\left(\mathbb{E}_{I}[F_H(\beta \, r^{\alpha}\,(I_t+N_0)|I_t)]\right)^{\tau})$. From the above result we know that $\Phi_c\subseteq \Phi_r$ is a nonhomogeneous PPP because its intensity $\lambda_c(r,\tau)$ is the intensity $\lambda_r$ of $\Phi_r$ scaled by \eqref{Eqn:ConnProbNtrans1}.
Since the channel gain $H$ is a Nakagami-\emph{m} random variable, we know that
\begin{eqnarray*}
\mathbb{E}_{I}\left[F_{H}(\beta r^{\alpha}I_t)|I_t)\right] &=& \int_{\mathbb{R}_+} F_H(\beta r^{\alpha}\omega)f_{I_t}(\omega)\, \dif \omega \nonumber\\ &\geq& 1-\int_{0}^{\infty}\frac{\Gamma(m,\beta r^{\alpha}\omega)}{\Gamma(m)}f_{I_t}(\omega)\,\dif \omega,
\end{eqnarray*}
where $\Gamma(m,x)=\int_{x}^{\infty} t^{m-1}e^{-t}\,\dif t$ and $\Gamma(m)=\Gamma(m,0)=(m-1)!$. Also, if $f_W(w)$ is a probability density function of random variable $W$ then we can have the following result.
\begin{eqnarray}
\int_{0}^{\infty} \Gamma(m,a w)\,f_W(w)\,\dif w = (-a)^m\frac{\dif ^{m-1}}{\dif a^{m-1}}\left(\frac{\mathcal{L}_{W}(a)}{-a}\right).
\end{eqnarray}
According to Lemma \ref{Lem:LapalceShotPPP}, we can have $\mathcal{L}_{I_t}(\phi)$ with $\Delta_1(\phi,\infty)$. Thus,
\begin{eqnarray}\label{Eqn:SuccProbNakaFading1}
&\mathbb{E}_{I}\left[F_{H}(\beta r^{\alpha}I_t|I_t))\right] \geq 1+\frac{(-\phi)^m}{\Gamma(m)}\frac{\dif^{m-1}}{\dif \phi^{m-1}}\left(\frac{\mathcal{L}_{I_t}(\phi)}{\phi}\right)\bigg|_{\phi=\beta r^{\alpha}}\nonumber \\
&\hspace{0.9in} = 1-(\beta r^{\alpha})\Psi^{(m-1)}(\beta r^{\alpha}).
\end{eqnarray}
Substituting \eqref{Eqn:SuccProbNakaFading1} into $\lambda_c(r,\tau)$, then \eqref{Eqn:NonCoopConnIntenNakaFading} can be arrived.

\subsection{Proof of Theorem \ref{Thm:MaxContenIntenNonCoop}}\label{App:ProofMaxContenIntenNonCoop}
According to the outage probability \eqref{Eqn:MultiOutProb2} upper bounded by $\epsilon$, we know
$$\mathbb{E}_{R}[\lambda_c(R,\tau)]\geq \lambda_r+\frac{\log(1-\epsilon)}{\mu_r}=\lambda_r\left(1-\frac{\epsilon}{k}\right)+\Theta(\epsilon^2),$$
for sufficiently small $\epsilon$. Calculating the lower bound of $\mathbb{E}_{R}[\lambda_c(R,\tau)]$ from \eqref{Eqn:NonCoopConnIntenNakaFading} and then using the H\"{o}lder inequality, we can show that
\begin{equation}\label{Eqn:AvgConnectProb}
\mathbb{E}_R[(\beta R^{\alpha})\Psi^{(m-1)}(\beta R^{\alpha})]\geq 1-\sqrt[\tau]{\epsilon/k}.
\end{equation}

Note that $(\beta\,R^{\alpha})\Psi^{(m-1)}(\beta\,R^{\alpha})\in(0,1)$ almost surely and thus its average approaches to unity when $\epsilon/k$ is sufficiently small such that $\sqrt[\tau]{\epsilon/k}\leq \epsilon$. That means $(\beta\,R^{\alpha})\Psi^{(m-1)}(\beta\,R^{\alpha})$ is nearly equal to one almost surely and thus $\lambda_t=\Theta(\epsilon)$. If $k$ is sufficiently large, then we have $\exp(-\mu_u\,\lambda_t\Delta_1)=1-\mu_u\,\lambda_t\Delta_1+\Theta(\epsilon^2)$. Substituting this expression into \eqref{Eqn:PsiMphi}, $(\beta\,R^{\alpha})\Psi^{(m-1)}(\beta R^{\alpha})$ can be reduced to $1- \mu_u\,\lambda_t\,\Delta_1\,\prod_{j=1}^{m-1}(1-\xi/j)+\Theta(\epsilon^2)$.
Define $\hat{\Delta}_1(\phi)\defn [\Delta_1(\phi,\infty)-\Delta_1(\phi,\hat{a}_B)]\prod_{j=1}^{m-1}(1-\xi/j)$. Choosing $\hat{a}_B\leq 1$, we have
\begin{eqnarray*}
\mathbb{E}_{R}\left[(1-(\beta\,R^{\alpha})\Psi^{(m-1)}(\beta\,R^{\alpha}))^{\tau}\right]\leq \frac{\left[\hat{\Delta}_1(\beta)\mu_r\lambda_t\right]^{\tau}}{\tau+1}+\Theta(\epsilon^2).
\end{eqnarray*}
Another upper bound for the above equation can be obtained by \eqref{Eqn:MultiOutProb2} and \eqref{Eqn:NonCoopConnIntenNakaFading}, and it should coincide with the above upper bound when $\lambda_t$ is maximized as shown in the following.
\begin{equation}\label{Eqn:MaxIntensityDenseNwks}
\bar{\lambda}_{\epsilon} = \frac{\sqrt[\tau]{\epsilon\,(\tau+1)}}{\mu_r\,\beta^{\xi}\,\sqrt[\tau]{k}\,\hat{\Delta}_1(\beta)}.
\end{equation}
So \eqref{Eqn:MaxContenIntenNoCoop} is obtained and the proof is complete.

% section of references
\bibliographystyle{ieeetran}
\bibliography{IEEEabrv,Ref_MTC}

\end{document}